\documentclass[10pt,aps,pra,twocolumn,floatfix]{revtex4-2}
\usepackage{graphicx} 
\usepackage{pgfplots}
\pgfplotsset{compat=1.18}
\usepackage[normalem]{ulem}
\usepackage{amsmath}
\usepackage{amssymb}
\usepackage{bbold}
\usepackage{amsthm}
\newtheorem{thm}{Theorem}
\newtheorem{lemma}[thm]{Lemma}

\usepackage{physics}
\usepackage{comment}
\usepackage{tikz}
\usetikzlibrary{decorations.pathreplacing}
\usetikzlibrary{intersections}
\usetikzlibrary{arrows.meta}
\usepgfplotslibrary{fillbetween}
\usepackage{quantikz}
\usepackage{xcolor}
\usepackage{physics}
\usepackage{bm}
\usepackage{booktabs}
\usepackage{makecell}
\usepackage{array, multirow}
\newcolumntype{P}[1]{>{\centering\arraybackslash}m{#1}}
\usepackage[colorlinks=true, linkcolor=black, citecolor=black, urlcolor=black]{hyperref}
\usepackage[capitalise, nameinlink]{cleveref}
\usepackage{natbib}
\newcommand{\aux}{\mathrm{aux}}
\newcommand{\anc}{\mathrm{anc}}
\newcommand{\hop}{\mathrm{hop}}
\newcommand{\init}{\mathrm{init}}

\newcommand{\phys}{\mathrm{phys}}
\newcommand{\ord}{\mathrm{ord}}

\newcommand{\Z}{\mathrm{Z}}
\newcommand{\X}{\mathrm{X}}
\newcommand{\Y}{\mathrm{Y}}
\newcommand{\Had}{\mathrm{H}}
\newcommand{\CZ}{\mathrm{CZ}}
\newcommand{\CX}{\mathrm{CX}}
\newcommand{\Id}{\mathbb{1}}
\begin{document}

\title{Efficient Simulation of Sparse, Non-Local Fermion Models}
\date{\today}
\author{Reinis Irmejs}
\email{reinis.irmejs@mpq.mpg.de}
\affiliation{Max-Planck-Institut für Quantenoptik, Hans-Kopfermann-Straße 1, D-85748 Garching, Germany}
\affiliation{Munich Center for Quantum Science and Technology (MCQST), Schellingstraße 4, D-80799 Munich, Germany}
\author{J. Ignacio Cirac}
\affiliation{Max-Planck-Institut für Quantenoptik, Hans-Kopfermann-Straße 1, D-85748 Garching, Germany}
\affiliation{Munich Center for Quantum Science and Technology (MCQST), Schellingstraße 4, D-80799 Munich, Germany}

\begin{abstract}
Efficient simulation of interacting fermionic systems is a key application of near-term quantum computers, but is hindered by the overhead required to encode fermionic operators on qubit hardware. Here, we consider models with $N$ fermionic modes in which each participates in at most a constant number $d$ of interactions and study the circuit depth required to implement the Trotterized time evolution on qubit hardware with all-to-all connectivity. We introduce an encoding that augments each physical fermionic mode with a small number of auxiliary fermions, enabling the removal of Jordan--Wigner strings. Although the preparation of the auxiliary fermion state incurs an initial overhead, this state remains invariant under time evolution. As a result, long-time evolution can be implemented with asymptotically optimal circuit depth, reducing a previously multiplicative $O(\log N)$ overhead to an additive overhead. Our results thus establish that the simulation of sparse fermionic models on qubit hardware matches the performance achievable on ideal fermionic hardware up to constant factors and $O(dN)$ ancillary qubits.
\end{abstract}

\maketitle
\section{Introduction}\label{sec:introduction}

Accurate and large scale simulation of interacting fermionic models is widely regarded as one of the earliest practically useful applications of near term quantum computers. Progress in this direction would not only deepen our understanding of fundamental physics but also support applied goals such as the design and characterization of novel materials \cite{Yudong_2019, Macardle_2020, Clinton_2024}. The two dimensional Fermi Hubbard model exemplifies this opportunity, as it is believed to capture essential features of high temperature superconductivity \cite{auerbach2012, Qin_2022}. Despite its apparent simplicity, the model remains poorly understood in the low temperature regime, where experimental access is limited \cite{mckay2011cooling, altman2021quantum, Greiner2025}. Recent advances in quantum computing platforms \cite{willow_2024, Helios_2025} and in digital fermion simulation \cite{Phasecraft_FH_2025} now point toward overcoming these limitations.

A central computational task in the simulation of fermionic models is the implementation of time evolution under a given Hamiltonian \cite{Gilyen_2019, Childs_2021}. Access to time evolution enables the estimation of correlation functions \cite{Bauer2016, Kreula2016, Kosugi2020, Cruz_2025}, the evaluation of finite energy properties through filtering techniques \cite{Lu_2021}, the approximation of ground state properties via adiabatic evolution \cite{messiah_quantum_2020, Jansen2007, Lidar2018}, and serves as a key component in variational algorithms \cite{Peruzzo_2014, Tilly_2022}. Simulation of time dynamics is believed to be difficult on classical computers \cite{Jordan2018}, emphasizing the need for efficient quantum computing algorithms.

The simulation of fermionic quantum many body systems on standard quantum computers requires encoding fermionic operators into qubit degrees of freedom. Implementing such encodings is challenging because fermionic operators satisfy canonical anticommutation relations that are not naturally captured by qubit operator algebras. To avoid this problem altogether, fermionic quantum computers have been proposed, in which fermionic statistics are implemented directly at the hardware level \cite{Bravyi2002_ferm}. While this approach eliminates the need for fermion to qubit encodings, it introduces substantial challenges for the implementation of quantum error correction, which remains significantly less developed than in qubit based architectures \cite{Pichler2024, Schukert2025}.

Motivated by these challenges, considerable effort has been devoted to optimizing fermionic encodings on qubit hardware. The simplest and most widely used mapping is the Jordan--Wigner transformation \cite{jordan1928paulische}, in which fermionic modes are ordered along a one dimensional chain. Fermionic anti commutation relations are enforced through non local string operators acting on all preceding sites, commonly referred to as Jordan--Wigner strings. As a consequence, fermionic interactions between sites $i$ and $j$ acquire a Pauli weight proportional to $|i-j|$. This leads to severe overheads for non local interactions, and already in two dimensional lattice models it is impossible to represent both horizontal and vertical nearest neighbor couplings with constant Pauli weight \cite{Guaita_2025}. More generally, models with non local fermionic terms result in encodings whose Pauli weight can scale extensively with system size.

The quality of more advanced fermionic encodings is naturally quantified by the circuit depth required to implement a single Trotter step of the time evolution. Achieving optimal depth requires both low Pauli weight of the encoded Hamiltonian terms and the ability to parallelize their implementation. For this reason, the Trotter step depth has been widely adopted as a benchmark for fermionic encodings \cite{Babbush_2018, Kivlichan2018, Clinton_2024, maskara2025, Childs2025}. To assess the optimality of a given encoding, this depth, together with the number of ancillary qubits, must be compared to the resources required to perform the same simulation directly on fermionic hardware with identical connectivity.

Whether a fermionic encoding can achieve constant overhead in circuit depth depends crucially on two factors: (i) the locality and sparsity of the Hamiltonian interactions and (ii) the geometry of the underlying quantum hardware. For local fermionic models in two or higher dimensions, numerous local to local encodings have been developed that achieve constant depth overhead on local qubit architectures, in some cases with additional error correcting structure \cite{verstraete2005mapping, setia2019superfast, jiang2019majorana, Derby_2021, chien2023simulatingquantumerrormitigation, landahl2023logicalfermionsfaulttolerantquantum, Parella2024, Obrien2024, Algaba_2024}. For all to all interacting fermionic models, such as Fermi--Hubbard type Hamiltonians, fermionic SWAP networks allow for optimal linear depth simulation on local qubit hardware \cite{Kivlichan2018}. The remaining scenario is the simulation of fermionic models with non local but sparse interactions on hardware with all to all connectivity. In this regime, each fermionic mode interacts with at most a constant number $d$ of other modes, as realized, for instance, in sparse instances of the SYK model or, more generally, in fermionic models defined on $d$-regular graphs.
 \cite{Tikhonov_2019, Garcia_2021, Herasymenko2023optimizingsparse, 
 Herre_2023, Tsironis_2024}.

Furthermore, such sparse and random models are believed to have ground states that can be efficiently prepared using quantum phase estimation, further emphasizing the need for efficient time-evolution algorithms \cite{Chen_2024}. On fermionic hardware with all-to-all connectivity, a single Trotter step for such models can be implemented in $\mathcal{O}(d)$ depth. We consider the simulation of this regime on qubit hardware with identical connectivity \cite{Helios_2025}. This regime can also be addressed using graph-based superfast encodings \cite{BRAVYI2002, setia2019superfast}. In this case, one needs $\mathcal{O}(dN)$ qubits, with the encoded fermionic terms having weight $\mathcal{O}(d)$ and $\mathcal{O}(\log d)$, respectively. The prohibitive cost of these approaches comes from the initialization of the stabilizer state, which requires $\mathcal{O}(N)$ circuit depth \cite{BRAVYI2002}.

Recent work has made substantial progress toward improving the dependence on the system size $N$. Using general fermionic permutations, it was shown that a single Trotter step can be implemented with depth $\mathcal{O}(d \log N)$ using an extensive number of ancillary qubits \cite{maskara2025, Childs2025}.
In this work, we build on these results and develop an approach that reduces the $\mathcal{O}(\log N)$ overhead from multiplicative to additive for long-time simulation.

Our method is inspired by \cite{verstraete2005mapping} and introduces a constant number of auxiliary fermions per physical site, which are used to construct stabilizer operators that eliminate Jordan--Wigner strings. To ensure equivalence with the original fermionic model, the auxiliary fermions are initialized in the joint $+1$ eigenspace of all stabilizers. The fermionic construction of the stabilizers allows us to use fermionic permutation techniques to initialize this state in $\mathcal{O}(\log N)$ depth. While preparing this state constitutes the dominant cost of the approach, it remains invariant under the subsequent time evolution. Indeed, since the preparation of the auxiliary fermion state is required only once, all subsequent Trotter steps can be carried out with constant overhead. As a result, if the total number of Trotter steps satisfies $M = \Omega(\log N)$, the total asymptotic cost of the simulation scales as $\mathcal{O}(M)$. In this regime, the total computational cost matches that of a fermionic quantum computer with the same connectivity.

In the next section, we introduce our approach and show how it maps fermionic operators to qubit operators with constant Pauli weight, together with the construction of the required auxiliary fermion initial state. In \cref{sec:time_dynamics}, we then summarize the full protocol and derive the total cost of the time evolution simulation. We conclude in \cref{sec:conclusions}.

\section{Method}

\subsection{Sparse, non-local Hamiltonian}

The objective of this work is to improve the overhead of performing long-time evolution of non-local fermion models. In this work, we consider Hamiltonians with sparse interactions and hoppings. To illustrate the idea, we consider a non-local hopping Hamiltonian:
\begin{equation}\label{eq:hop_on_graph}
    H_{\mathrm{hop}}
    = \sum_{(i,j)\in E} t_{ij}\,(a_i^\dagger a_j + a_j^\dagger a_i).
\end{equation}
 In \cref{app:JW_transformations} we show how to use our transformation for general fermion interactions and show how to simulate sparse Fermi--Hubbard and SYK models in \cref{app:example_models}. To encode the connectivity pattern, we use a graph $\mathcal{G} = (V,E)$: we assume that $(i,j)\in E$ exactly when the hopping term between sites $i$ and $j$ is present ($t_{ij}\neq 0$). 

To illustrate the problem that arises for non-local models consider a hopping term between sites $i$ and $j$ ($i<j$) under the Jordan--Wigner (JW) transformation \cite{jordan1928paulische}:
\begin{equation}
    a_i^\dagger a_j + a_j^\dagger a_i = \frac{1}{2}(\X_i\X_j+\Y_i\Y_j)\prod_{k=i}^{j-1}\Z_k.
\end{equation}
For sites $i$ and $j$ that are far apart, the problem comes from the term $\prod_{k=i}^{j-1}\Z_k,$ which has support on all sites between them. This leads to a high cost when implementing time-evolution since it highly restricts the number of terms we can implement in parallel. 

The main aim of our work is to perform a transformation on each Hamiltonian term and map it to a qubit term that has a constant Pauli weight. We consider a transformation of form:
\begin{align}
\label{eq:trans_intro}
    &a_i^\dagger a_j + a_j^\dagger a_i \xrightarrow[]{} (a_i^\dagger a_j + a_j^\dagger a_i)P_{ij},\\&\text{where it is has support on $i,j$ only.}\nonumber
\end{align}
Note that for local Hamiltonians on local hardware this corresponds to \cite{verstraete2005mapping}. Combined with access to a non-local hardware this will allow us to devise a scheme to perform the time-evolution in an asymptotically efficient way. In the next sections, we will show how to construct the transformation based on the stabilizer operators $P_{ij}$, as well as, how to enforce that the physics remain unchanged.

\subsection{Setup}
\label{sec:fermion_setup}

We consider a system of $N$ sites, labeled $i = 1, \dots, N$. On every site we have one physical fermionic mode $a_i$. In addition, for every site $i$ we introduce $\nu$ auxiliary fermions $b^{(\ell)}_i$, where $1\leq \ell \leq \nu$. We assume a JW encoding where the fermionic modes are ordered as in \cref{fig:fermionic-lattice}: 
\begin{equation*}
    \bigl( a_1, b^{(1)}_1,\dots,b^{(\nu)}_1,
    a_2, b^{(1)}_2,\dots,b^{(\nu)}_2,
    \dots,
    a_N, b^{(1)}_N,\dots,b^{(\nu)}_N
    \bigr).
\end{equation*}
 We associate a qubit with each fermionic mode, and we denote the Pauli operators acting on the qubit corresponding to mode type $(\ell)$ at site $i$ by $\X_i^{(\ell)}$, $\Y_i^{(\ell)}$, and $\Z_i^{(\ell)}$. The physical modes $a_i$ are associated with the superscript $(0)$.
 Under the JW transformation, the fermionic operators are:
 \begin{align}
    a_i
     &= S_i\,\frac{\X^{(0)}_i - i\Y^{(0)}_i}{2},
    \label{eq:JW_a_i}\\
    b_i^{(\ell)}
     &= S_i\,Z_{i,<\ell} \frac{\X^{(\ell)}_i - i\Y^{(\ell)}_i}{2},\quad \text{where}\\
     S_i&= \prod_{j=1}^{i-1} \prod_{k=0}^{\nu} \Z_j^{(k)}\quad\text{and}\quad Z_{i,<\ell}
    = \prod_{m=0}^{\ell-1} \Z_i^{(m)}. 
    \label{eq:JW_b_i_l}
\end{align}
We introduce $S_i$ and $Z_{i,<\ell}$ to simplify the notation of the \emph{JW strings}. 
\begin{figure}[t]
  \centering
  \begin{tikzpicture}[
      x=1.4cm,
      y=1.1cm,
      site/.style={circle, draw, minimum size=8mm, inner sep=0pt},
      phys/.style={site, draw=cyan!60!black, fill=cyan!10},
      aux/.style={site, draw=olive!60!black, fill=olive!10},
      jwline/.style={dashed, thin, opacity=0.75},
      >=stealth
    ]

    \def\yone{-0.5}
    \def\ytwo{-1.5}
    \def\ythree{-2.5}
    \def\ydots{-3.5}
    \def\yn{-4.5}

    \def\ymidupper{-3.0}  
    \def\ymidlower{-4.0}  


    \draw[rounded corners=6pt, fill=olive!2, draw=olive!40!black]
    (1.65,\yone+1.0) rectangle (2.35,\yn-0.5);

    \draw[jwline]
      (0,\yone) -- (1,\yone) -- (2,\yone) -- (2.5,\yone);
    \draw[jwline]
      (3.5,\yone) -- (4,\yone);

    \draw[jwline]
      (4,\yone) -- (4,\yone-0.5) -- (0,\yone-0.5) -- (0,\ytwo);

    \draw[jwline]
      (0,\ytwo) -- (1,\ytwo) -- (2,\ytwo) -- (2.5,\ytwo);
    \draw[jwline]
      (3.5,\ytwo) -- (4,\ytwo);

    \draw[jwline]
      (4,\ytwo) -- (4,\ytwo-0.5) -- (0,\ytwo-0.5) -- (0,\ythree);

    \draw[jwline]
      (0,\ythree) -- (1,\ythree) -- (2,\ythree) -- (2.5,\ythree);
    \draw[jwline]
      (3.5,\ythree) -- (4,\ythree);

    \draw[jwline]
      (4,\ythree) -- (4,\ymidupper) -- (0,\ymidupper)--(0, \ymidupper-0.1);


    \draw[jwline]
      (0,\ymidlower) -- (0,\yn);

    \draw[jwline]
      (0,\yn) -- (1,\yn) -- (2,\yn) -- (2.5,\yn);
    \draw[jwline]
      (3.5,\yn) -- (4,\yn);


    \node[anchor=south] at (0,0.7) {\shortstack{physical\\modes}};
    \node[anchor=south] at (2.5,0.7) {\shortstack{auxiliary\\modes}};
    \node[anchor=south] at (1,0.1) {$\ell=1$};
    \node[anchor=south] at (2,0.1) {$\ell=2$};
    \node[anchor=south] at (3,0.1) {$\cdots$};
    \node[anchor=south] at (4,0.1) {$\ell=\nu$};

    \node[anchor=east] at (-0.6,\yone)   {$i=1$};
    \node[anchor=east] at (-0.6,\ytwo)   {$i=2$};
    \node[anchor=east] at (-0.6,\ythree) {$i=3$};
    \node[anchor=east] at (-0.6,\ydots)  {$\vdots$};
    \node[anchor=east] at (-0.6,\yn)     {$i=N$};

    \node[phys] at (0,\yone)   {$a_{1}$};
    \node[phys] at (0,\ytwo)   {$a_{2}$};
    \node[phys] at (0,\ythree) {$a_{3}$};
    \node        at (0,\ydots) {$\vdots$};
    \node[phys] at (0,\yn)     {$a_{N}$};

    \node[aux] at (1,\yone) {$b_{1}^{(1)}$};
    \node[aux] at (2,\yone) {$b_{1}^{(2)}$};
    \node      at (3,\yone) {$\cdots$};
    \node[aux] at (4,\yone) {$b_{1}^{(\nu)}$};

    \node[aux] at (1,\ytwo) {$b_{2}^{(1)}$};
    \node[aux] at (2,\ytwo) {$b_{2}^{(2)}$};
    \node      at (3,\ytwo) {$\cdots$};
    \node[aux] at (4,\ytwo) {$b_{2}^{(\nu)}$};

    \node[aux] at (1,\ythree) {$b_{3}^{(1)}$};
    \node[aux] at (2,\ythree) {$b_{3}^{(2)}$};
    \node      at (3,\ythree) {$\cdots$};
    \node[aux] at (4,\ythree) {$b_{3}^{(\nu)}$};

    \node at (1,\ydots) {$\vdots$};
    \node at (2,\ydots) {$\vdots$};
    \node at (3,\ydots) {$\ddots$};
    \node at (4,\ydots) {$\vdots$};

    \node[aux] at (1,\yn) {$b_{N}^{(1)}$};
    \node[aux] at (2,\yn) {$b_{N}^{(2)}$};
    \node      at (3,\yn) {$\cdots$};
    \node[aux] at (4,\yn) {$b_{N}^{(\nu)}$};

\end{tikzpicture}
  \caption{Physical modes $a_i$ (blue) and auxiliary modes $b_i^{(\ell)}$ (olive) together with the JW ordering path (dashed). The rectangular selection encompasses one whole auxiliary mode family.}
  \label{fig:fermionic-lattice}
\end{figure}
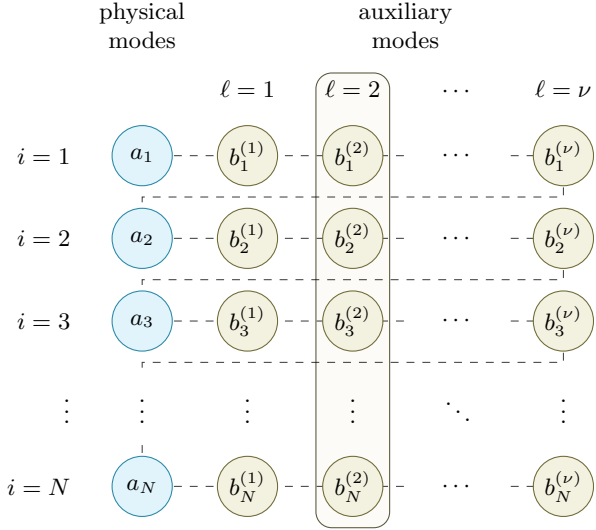
In order to construct the transformation, we will use the Majorana operators associated with the above auxiliary modes $b_i^{(\ell)}$. These are simply defined as:
\begin{align}
    c_i^{(\ell)} &= b_i^{(\ell)} + b_i^{\dagger(\ell)},
    \label{eq:def_c}\\
    d_i^{(\ell)} &= -i(b_i^{(\ell)} - b_i^{\dagger(\ell)}),
    \label{eq:def_d}
\end{align}
and they obey the canonical Majorana algebra:
\begin{align}
    \{c_i^{(k)}, c_j^{(\ell)}\} &= \{d_i^{(k)}, d_j^{(\ell)}\} = 2\delta_{ij}\delta_{k\ell},\;\;
    \{c_i^{(k)}, d_j^{(\ell)}\} = 0.
    \label{eq:majorana_id}
\end{align}
When mapping Hamiltonian terms to qubits, we will use Pauli weight $w$ to denote the number of qubits on which a term acts non-trivially.

\subsection{Constructing the encoding}
\label{sec:contructing_encoding}

The key idea in constructing this transformation is to create a fermionic, hermitian operator $P^{(\ell)}_{ij}$ from the auxiliary fermion modes possessing the same \emph{JW string}. Subsequently, by transforming the initial Hamiltonian terms as in \eqref{eq:trans_intro}, we will be able to eliminate the extensive \emph{JW strings}, resulting in a term that has a Pauli weight independent of the distance between the two interacting modes $\abs{i-j}$. However, such transformation does not preserve the physics in the full Hilbert space. To achieve that, we require that the auxiliary fermion modes are prepared in the $+1$ eigenstate of $P^{(\ell)}_{ij}$. We will refer to $P^{(\ell)}_{ij}$ as stabilizers. Thus, the main two challenges of this construction are:
\begin{enumerate}
    \item constructing the stabilizers $P^{(\ell)}_{ij}$, such that they can eliminate the \emph{JW strings} for all present interaction terms in the initial Hamiltonian,
    \item efficiently preparing the initial state on the auxiliary modes $\ket{\Psi_{\aux}}$, such that:
    \begin{equation}
    \label{eq:P_stabilizer}
        P^{(\ell)}_{ij} \ket{\Psi_{\aux}} = +\ket{\Psi_{\aux}}
    \end{equation}
    for every relevant stabilizer $P_{ij}^{(\ell)}$.
\end{enumerate}
The condition \cref{eq:P_stabilizer} asserts that the physics of the initial Hamiltonian remain unchanged, and the auxiliary state remains invariant:
\begin{equation}
\label{eq:enforce_physics}
    e^{-i\tau h_{ij}P_{ij}^{(\ell)}}\ket{\Psi_{\aux}}\ket{\Psi_{\phys}} = \ket{\Psi_{\aux}}e^{-i\tau h_{ij}}\ket{\Psi_{\phys}}.
\end{equation}

\paragraph{Constructing the stabilizers.} Firstly, to construct the stabilizers we use the Majorana operators from \cref{eq:def_c,eq:def_d}. A convenient way is to do it similarly as in \cite{verstraete2005mapping}:
\begin{equation}
    \label{eq:Pij_def}
    P_{ij}^{(\ell)} = ic_i^{(\ell)}d_j^{(\ell)} = -iS_iZ_{i,<\ell}X_i^{(\ell)}S_jZ_{j,<\ell}Y_j^{(\ell)}
\end{equation}
One can easily check that $P_{ij}^{(\ell)}$ is hermitian and $P_{ij}^{(\ell)2}=1$, thus it has eigenvalues of $\pm1$ (see \cref{app:stabilizer_properties}). By using the Majorana operators on sites $i,j$ we ensure that the stabilizer will have the same \emph{JW string} we wish to simplify. As an example, consider the transformation of the aforementioned hopping term from  \cref{eq:trans_intro}. Under the JW transform:
\begin{align}
    \label{eq:trans_JW_example}
    &(a_i^\dagger a_j +a_j^\dagger a_i)P_{ij}^{(\ell)}\xrightarrow[]{}
    \frac{1}{2}(\X^{(0)}_i\X^{(0)}_j+\Y^{(0)}_i\Y^{(0)}_j)S_i S_j \times\nonumber \\
    &\times (-iS_iZ_{i,<\ell}\X_i^{(\ell)}S_jZ_{j,<\ell}\Y_j^{(\ell)}) \nonumber\\&= \frac{i}{2}(\X^{(0)}_i\X^{(0)}_j+\Y^{(0)}_i\Y^{(0)}_j)Z_{i,<\ell}Z_{j,<\ell}\X_i^{(\ell)}\Y_j^{(\ell)}.
\end{align}
From the transformation, one can see that the Pauli weight of the transformed hopping term is $w_{\mathrm{hop}} = 2\ell+2$. In \cref{app:JW_transformations} we provide a detailed description of how any even fermionic term can be transformed to a local qubit operator by these stabilizers.

\paragraph{Stabilizer initial state.} Secondly, to ensure that the Hamiltonian evolution remains unchanged under the transformation we need to ensure that $P_{ij}^{(\ell)}$ is a $+1$ eigenstate of $\ket{\Psi_{\aux}}$. For a single stabilizer that is given by:
\begin{equation}
\label{eq:single_P_eigenstate}
    P_{ij}^{(\ell)}\frac{c^{(\ell)}_i-id^{(\ell)}_j}{\sqrt{2}}\ket{0_{\aux}} = \frac{c^{(\ell)}_i-id^{(\ell)}_j}{\sqrt{2}}\ket{0_{\aux}}.
\end{equation}
However, any Hamiltonian of interest will have interactions between multiple modes, hence we have to ensure that all associated stabilizers are the mutual $+1$ eigenstates of $\ket{\Psi_{\aux}}$. A necessary condition to ensure that is that all stabilizers mutually commute. Since there can be several unique interactions per each lattice mode, our construction uses multiple ($\nu$) ancillary registers per each physical site, which allows to group them in a commuting way. In the next section, we show how to construct the auxiliary initial state for an arbitrary Hamiltonian.

\subsection{Preparing the full auxiliary state}
\label{sec:aux_state_preparation}

\paragraph{Grouping the stabilizers.} First, to ensure that all stabilizers commute, we need to group them appropriately between the auxiliary registers. In the following, we present an optimal grouping strategy based on graph coloring that ensures that we need to use the minimum number $\nu$ of ancillary qubits per site. For simplicity, here we consider a particular example of an interaction graph as given in \cref{fig:coloring_example}~(a), while a rigorous treatment is given in \cref{app:stabilizer_grouping}.
\begin{figure}
    \centering
    \begin{tikzpicture}[
    scale=2,
    every node/.style={circle, draw, thick, minimum size=6mm, inner sep=0pt, fill = gray!10},
    edgecyan/.style={line width=2pt, cyan!80!black},
    edgeorange/.style={line width=2pt, orange!80!black},
    edgeolive/.style={line width=2pt, olive!80!black},
    edgepurple/.style={line width=2pt, magenta!70!black},
    labeltext/.style={sloped, above, inner sep=0pt, outer sep=0pt, draw=none, fill=none,shape=rectangle, font=\normalsize},
    arrowstyle/.style={-{Stealth[length=10pt]}}
]

\begin{scope}
\node[labeltext] at (-0.8,0.95) {(a)};
  \node (v1) at (-0.7,0.35)   {1};
  \node (v2) at (0,1.05)      {2};
  \node (v3) at (0.7,0.35)    {3};
  \node (v4) at (0.7,-0.35)   {4};
  \node (v5) at (0,-1.05)     {5};
  \node (v6) at (-0.7,-0.35)  {6};

  \node (v7) at (0,-0.35)     {7};
  \node (v8) at (0,0.35)      {8};

  \draw[thick] (v1) -- (v2) -- (v3) -- (v4) -- (v5) -- (v6) -- (v1);

  \draw[edgecyan] (v5) -- (v6);
  \draw[edgecyan] (v4) -- (v7);
  \draw[edgecyan] (v2) -- (v8);

  \draw[edgeolive] (v4) -- (v5);
  \draw[edgeolive] (v6) -- (v7);
  \draw[edgeolive] (v2) -- (v1);
  \draw[edgeolive] (v3) -- (v8);

  \draw[edgeorange] (v1) -- (v8);
  \draw[edgeorange] (v3) -- (v4);
  \draw[edgeorange] (v5) -- (v7);

  \draw[edgepurple] (v1) -- (v6);
  \draw[edgepurple] (v2) -- (v3);
  \draw[edgepurple] (v7) -- (v8);
\end{scope}


\begin{scope}[xshift=2cm]
    \draw[rounded corners=6pt, fill=olive!2, draw=olive!40!black]
        (-1,1.3) rectangle (1,-1.3);
    \node[labeltext] at (-0.8,0.95) {(b)};
  \node (w1) at (-0.7,0.35)   {1};
  \node (w2) at (0,1.05)      {2};
  \node (w3) at (0.7,0.35)    {3};
  \node (w4) at (0.7,-0.35)   {4};
  \node (w5) at (0,-1.05)     {5};
  \node (w6) at (-0.7,-0.35)  {6};
  \node (w7) at (0,-0.35)     {7};
  \node (w8) at (0,0.35)      {8};

  \draw[edgecyan,arrowstyle] (w5) -- node[labeltext] {$P_{56}$} (w6) ;
  \draw[edgecyan,arrowstyle] (w7) -- node[labeltext] {$P_{74}$} (w4);
  \draw[edgecyan,arrowstyle] (w2) -- node[labeltext] {$P_{28}$} (w8);

  \draw[edgeolive,arrowstyle] (w4) -- node[labeltext] {$P_{45}$} (w5);
  \draw[edgeolive,arrowstyle] (w6) -- node[labeltext] {$P_{67}$} (w7);
  \draw[edgeolive,arrowstyle] (w1) -- node[labeltext] {$P_{12}$} (w2);
  \draw[edgeolive,arrowstyle] (w8) -- node[labeltext] {$P_{83}$} (w3);

\end{scope}

\begin{scope}[yshift=-2.7cm]
    \draw[rounded corners=6pt, fill=cyan!2, draw=cyan!40!black]
        (-1,1.3) rectangle (3,0.65);
    \node[labeltext] at (-0.8,0.95) {(c)};
    \node[labeltext, font=\small] at (1,0.95) {$\mathcal{C}^{(1)} \propto {\color{cyan!80!black}(c^{(1)}_2-id^{(1)}_8)(c^{(1)}_7-id^{(1)}_4)(c^{(1)}_5-id^{(1)}_6)}\times$};
    \node[labeltext, font=\small] at (1,0.7) {$\times{\color{olive!80!black}(c^{(1)}_1-id^{(1)}_2)(c^{(1)}_8-id^{(1)}_3)(c^{(1)}_6-id^{(1)}_7)(c^{(1)}_4-id^{(1)}_5)}$};
\end{scope}

\begin{scope}[yshift=-3.45cm]
    \draw[rounded corners=6pt, fill=magenta!2, draw=magenta!40!black]
        (-1,1.3) rectangle (3,0.7);
    \node[labeltext] at (-0.8,0.95) {(d)};
    \node[labeltext, font=\small] at (1,0.95) {${\color{olive!80!black}\sigma_{\mathrm{green}}(1,2,8,3,6,7,4,5) = (1,2,3,4,5,6,7,8)},$};
    \node[labeltext, font=\small] at (1,0.75) {${\color{cyan!80!black}\sigma_{\mathrm{blue}}(2,8,7,4,5,6,1,3) = (1,2,3,4,5,6,7,8)}.$};
\end{scope}
\end{tikzpicture}
    \caption{(a) Proper edge coloring of the interaction graph with $N=8$ vertices.
    (b) Subgraph of green and blue edges. Arrows denote the assigned edge orientations. With this ordering, all depicted stabilizers $P_{ij}$ commute and can be encoded in the same auxiliary register $\ell$.
    (c) Operator preparing the initial state on auxiliary modes $\ell=1$, associated with the blue and green stabilizers.
    (d) Permutations relevant for stabilizers associated with the blue and green colors, bringing them into an ordered form.
}
    \label{fig:coloring_example}
\end{figure}
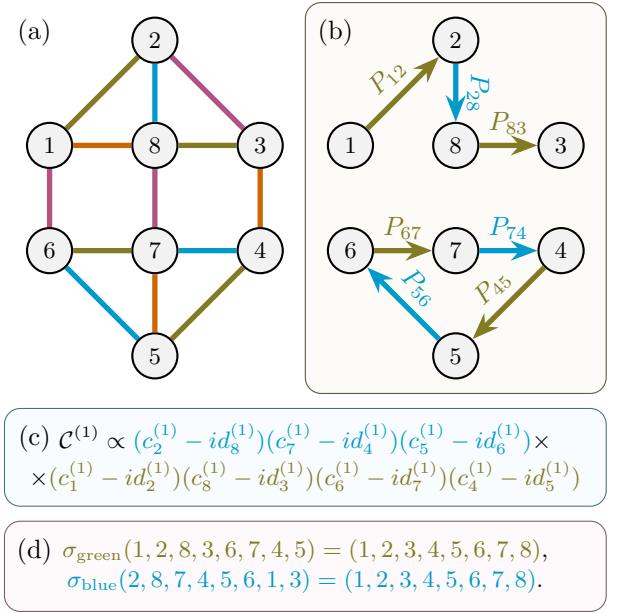
\cref{fig:coloring_example}~(a) depicts an interaction graph $\mathcal{G}$ for a hopping Hamiltonian as in \cref{eq:hop_on_graph}. In particular, it means that the model has a hopping interaction between modes $i,j$ only if the graph $\mathcal{G}$ has an edge connecting the respective vertices. 

In \cref{sec:contructing_encoding}, we saw that a necessary requirement for the initial state preparation is to appropriately group the stabilizers between the auxiliary modes, such that they all mutually commute. We achieve such grouping using the proper graph edge coloring - a coloring of edges such that each vertex only has edges of distinct colors. In particular, we are interested in the minimum number of colors required to achieve this - the chromatic index of a graph $\chi(\mathcal{G})$. In our example, the chromatic index is 4 (which is lower bounded by the maximum degree of any vertex). Note that the stabilizers associated with the edges of the same color trivially commute since they are supported on distinct vertices. Furthermore, in \cref{fig:coloring_example}~(b) we illustrate that under a correct ordering, the stabilizers of the edges of any two colors can be grouped together in a commuting set. Thus, in this example we require $\nu=2$ ancillary registers per site to encode all stabilizers, or more generally:
\begin{equation}
    \nu = \left \lceil{\frac{\chi(\mathcal{G})}{2}} \right \rceil.
\end{equation}
We denote a (now correctly ordered) stabilizer edge group associated with a color $\gamma$ (e.g. all the edges in blue) as:
\begin{equation}
    E_\gamma = \{(i_1,j_1),\dots,(i_{\Gamma_\gamma},j_{\Gamma_\gamma})\},
\end{equation}
where $\Gamma_{\gamma} = \abs{E_{\gamma}}$.

\paragraph{Full auxiliary state.} Second, given an appropriate stabilizer grouping, we discuss how to prepare the initial state. We define $\mathcal{C}^{(\ell)}$ as the operator for preparing the initial state associated with stabilizers on auxiliary modes $\ell$. We choose to group on the auxiliary register $\ell$ stabilizers associated with edges $E_{2\ell-1}$ and $E_{2\ell}$ (also see \cref{fig:coloring_example}~(c)):
\begin{align*}
    \mathcal{C}^{(\ell)}
     &= \left(\prod_{(i,j)\in E_{2\ell}} \frac{c_i^{(\ell)} - i d_j^{(\ell)}}{\sqrt{2}}\right)\!\!\!
    \left(\prod_{(i',j')\in E_{2\ell-1}} \frac{c_{i'}^{(\ell)} - i d_{j'}^{(\ell)}}{\sqrt{2}}\right)\!.
    \label{eq:C_l_def}
\end{align*}
Furthermore, the full auxiliary initial state is then:
\begin{equation}
    \ket{\Psi_{\aux}} = \prod_{\ell=1}^{\nu} \mathcal{C}^{(\ell)}\ket{0_{\aux}}.
\end{equation}

\paragraph{Permutation to an ordered state.} Third, to ensure an efficient preparation of the initial state, we exploit the generalized fermion permutations from \cite{maskara2025, Childs2025}. The aim here is to permute the arbitrary ordering from any $E_{\gamma}$ into one where the modes are ordered increasingly - $\{(1,2), (3,4), \dots, (2\Gamma_{\gamma-1}, 2\Gamma_{\gamma})\}$. Having the modes in this ordered form eliminates extensive \emph{JW strings} between the $c$ and $d$ Majorana operators in each term, allowing for an efficient state preparation routine that is independent of the system size $N$. For each color we define the permutation (also see \cref{fig:coloring_example}~(d)):
\begin{equation}
    \sigma_\gamma(i_1,j_1,\dots,i_{\Gamma_\gamma},j_{\Gamma_\gamma}) = (1,2,\dots,2\Gamma_\gamma).
\end{equation}
When writing this permutation it is implied that all $\nu$ fermion modes associated with site $i$ get permuted, while keeping their intrinsic order. In both \cite{maskara2025, Childs2025} it is shown how to implement the unitary $U_{\sigma_{\gamma}}$ associated with the permutation $\sigma_{\gamma}$ in $O(\log N)$ depth \cite{note_perm}.

\paragraph{Ordered state preparation.} Fourth, the last step in the initial state preparation routine is to show how to implement the ordered operator:
\begin{equation}
\label{eq:C_ord}    \mathcal{C}^{(\ell)}_{\mathrm{ord}} = \frac{1}{2^{\frac{N}{4}}}(c_1^{(\ell)}-id_2^{(\ell)})(c_3^{(\ell)}-id_4^{(\ell)})\times \dots \times (c_{N-1}^{(\ell)}-id_{N}^{(\ell)}).
\end{equation}
Having access to a circuit implementing this operation then allows us to prepare the full auxiliary initial state as:
\begin{equation}
    \ket{\Psi_{\aux}} = \prod_{\ell=1}^{\nu} \left (U_{\sigma_{2\ell-1}}\mathcal{C}^{(\ell)}_{\mathrm{ord}}U^{\dagger}_{\sigma_{2\ell-1}}U_{\sigma_{2\ell}}\mathcal{C}^{(\ell)}_{\mathrm{ord}}U^{\dagger}_{\sigma_{2\ell}}\right )\ket{0_{\aux}},
\end{equation}
where the permutation products $U^{\dagger}_{\sigma_{2\ell-1}}U_{\sigma_{2\ell}}$ can be further simplified into a single permutation. We discuss the specific circuit for the ordered state preparation $\mathcal{C}^{(\ell)}_{\mathrm{ord}}$ in \cref{app:ordered_state_prep}, where we show that:
\begin{equation}
    \mathrm{Depth}(\mathcal{C}^{(\ell)}_{\mathrm{ord}}) = \mathcal{O}(\log(\nu)),
\end{equation}
along with $N/2$ additional auxilliary qubits.

\section{Summary of the Time Dynamics Simulation}
\label{sec:time_dynamics}

In this section, we summarize the protocol for encoding sparse, non-local fermionic models and derive the cost of Trotterized time evolution. We consider a Hamiltonian $H_{\hop}$ (\cref{eq:hop_on_graph}) defined on a connectivity graph $\mathcal{G}$ with chromatic index $\chi$, which determines both the number of Trotter layers and the number of auxiliary modes per site. To implement the transformation, we introduce $\nu = \lceil \chi/2 \rceil$ auxiliary fermion modes per site. We first transform the physical Hamiltonian $H \rightarrow \tilde{H}$ and prepare the auxiliary modes in a state $\ket{\Psi_{\aux}}$. This construction ensures that, for all evolution times $T$,
\begin{equation*}
    e^{-i\tilde{H}T}\ket{\Psi_{\aux}}\ket{\Psi_{\phys}} = \ket{\Psi_{\aux}}e^{-iHT}\ket{\Psi_{\phys}}.
\end{equation*}
In \cref{tbl:scaling_ingredients} we summarize the cost of performing each step of the algorithm. The total circuit depth required for the transformation and preparation of $\ket{\Psi_{\aux}}$ is given by:
\begin{align}
    \mathrm{Depth}(\ket{\Psi_{\aux}}) &= 2\nu(\mathrm{Depth}(\mathcal{C}^{(\ell)}_{\mathrm{ord}})+\mathrm{Depth}(U_{\sigma_\gamma}))\nonumber\\
    &=\mathcal{O}\left(\chi\log(\chi)+\chi\log(N)\right)
\end{align}
and the approach requires $(2\nu+1) N$ ancillary qubits \cite{note_ancillas}.

\begin{table}[t]
    \centering
    \caption{Resource scaling of the individual steps of the algorithm. Fermion permutation constructions follow~\cite{maskara2025,Childs2025}.}
    \label{tbl:scaling_ingredients}
    \begin{tabular}{P{4.5cm} P{1.8cm} P{1.8cm}}
    \toprule
    Operation & Depth & Ancillas \\
    \midrule
    \multicolumn{3}{c}{\textsc{State preparation}} \\
    \addlinespace
    Introducing auxiliary fermions & $0$ & $\lceil\chi/2\rceil\,N$ \\
    Ordered state preparation $\mathcal{C}_{\ord}^{(\ell)}$ 
        & $\mathcal{O}(\log\chi)$ & $N/2$ \\
    Fermion permutation $U_{\sigma_\gamma}$ (with measurement)
        & $\mathcal{O}(\log(N))$ & $(\lceil\chi/2\rceil+1)N$ \\
    Fermion permutation $U_{\sigma_\gamma}$ (without measurement)
        & $\mathcal{O}(\log^2(\chi N))$ & $0$ \\
    \addlinespace
    \multicolumn{3}{c}{\textsc{Time evolution}} \\
    \addlinespace
    One-color Trotter layer $e^{-i\tau\tilde{h}_\alpha}$ 
        & $\mathcal{O}(\log\chi)$ & $0$ \\
    Full Trotter layer $\exp(-i\tilde{H}\tau)$ 
        & $\mathcal{O}(\chi\log\chi)$ & $0$ \\
    \bottomrule
    \end{tabular}
\end{table}

We approximate the time evolution using a first-order Trotter formula, and consider the cost of implementing $M$ applications of the unitary $W$:
\begin{align}
    W & = e^{-i \tilde{H} \tau} =\prod_{\alpha=1}^{\chi} e^{-i\tau\tilde{h}_\alpha} + \mathcal{O}( \Lambda  \tau^2)\\
\end{align}
where $\Lambda$ is a bound on the spectral norm of commutators, and scales as $\mathcal{O}(\chi^2N)$ (see \cref{app:Trotter_errors}) \cite{Childs_2021}.

\paragraph{Depth per Trotter step.}
To implement a single product $\prod_{\alpha=1}^{\chi} e^{-i\tau \tilde{h}_\alpha}$ on qubit hardware, we first note that, within a given $\tilde{h}_\alpha$, no two interaction terms overlap and can therefore be implemented in parallel. Moreover, each term in $\tilde{h}_\alpha$ has Pauli weight $O(\chi)$, and can thus be implemented in depth $\mathcal{O}(\log \chi)$ using the Pauli gadget technique \cite{Remaud_2025} (see \cref{sec:Pauli_gadget}). As a result,
\begin{equation}
    \mathrm{Depth}(e^{-i\tau \tilde{h}_\alpha}) = \mathcal{O}(\log(\chi)).
\end{equation}
Since the Hamiltonian is decomposed into $\chi$ sequential layers, the depth per Trotter step is $\mathrm{Depth}_{\mathrm{step}} = \mathcal{O}(\chi\log(\chi))$, yielding a total depth for implementing $M$ steps of Trotterized evolution of
\begin{align}
    \mathrm{Depth}(W^M) &= \mathrm{Depth}_{\init} + M\mathrm{Depth}_{\mathrm{step}}\\
    &=\mathcal{O}(\chi\log(\chi) + \chi\log( N)+M\chi\log(\chi))\nonumber.
\end{align}
\paragraph{Trotter Error.} Performing $M$ Trotter steps incurs an error of order $\mathcal{O}(\chi^2 N M \tau^2) = \mathcal{O}(\chi^2 N T^2 / M)$, where the total evolution time is $T = M\tau$. To ensure that the Trotter error remains small, it suffices to choose $M = \mathcal{O}(\chi^2 N T^2)$ steps for first-order Trotter evolution. By appropriately adjusting the value of $\tau$ in $W$, the same construction can be extended to higher-order Trotter formulae \cite{Childs_2021}. Finally, the value of $\Lambda$ is unchanged by our transformation (\cref{app:Trotter_errors}), yielding the same error as on fermionic hardware.

\paragraph{Asymptotic optimality.}
On ideal fermionic hardware, $M$ Trotter steps can be implemented with circuit depth
\begin{align}
    \mathrm{Depth}_{\mathrm{fermion}}(W^M) = \mathcal{O}(M\chi),
\end{align}
since the Hamiltonian must still be decomposed into $\chi$ non-overlapping terms. In the regime $M \gg \log N$, the ratio of the circuit depth on qubit hardware to that on fermionic hardware satisfies
\begin{equation}
    \frac{\mathrm{Depth}(W^M)}{\mathrm{Depth}_{\mathrm{fermion}}(W^M)} = \mathcal{O}(\log \chi),
\end{equation}
demonstrating that our encoding achieves the same asymptotic scaling for long-time dynamics as a fermionic quantum computer, up to a constant factor.

\section{Conclusions}
\label{sec:conclusions}

In this work, we established asymptotically optimal circuit depth for the simulation of sparse non-local fermionic Hamiltonians on qubit hardware. By introducing auxiliary fermionic modes and defining stabilizers that eliminate Jordan-Wigner strings from all even interaction terms, we map the Hamiltonian to one with terms of constant Pauli weight. The costliest component of our approach is the initialization of the auxiliary fermions in the mutual $+1$ eigenstate of all stabilizers, which can be achieved in $\mathcal{O}(\log(N))$ depth. While a single Trotter step therefore requires $\mathcal{O}(\log(N))$ depth, matching previous approaches, the auxiliary state remains invariant under time evolution, allowing each subsequent step to be implemented with depth independent of the system size. These results provide a practical route to efficient quantum simulation of sparse fermionic models and close the open gap by showing that fermionic quantum computation offers no asymptotic advantage in circuit depth, but only in requiring a constant number of ancillas per site.

\acknowledgments{We thank Sirui Lu for helpful discussions. We
acknowledge the support from the German Federal Ministry of Education and Research (BMBF) through FermiQP (Grant No. 13N15890) within the funding program quantum technologies—from basic research to market. This research is
part of the Munich Quantum Valley (MQV), which is supported by the Bavarian state government with funds from the Hightech Agenda Bayern Plus.
This work was partially supported by the Deutsche Forschungsgemeinschaft (DFG, German Research Foundation) under Germany's Excellence Strategy -- EXC-2111 -- 390814868; and by the EU-QUANTERA project TNiSQ (BA 6059/1-1).}
\medskip

\bibliography{main}

\appendix

\section{Detailed JW transformations}\label{app:JW_transformations}

In this appendix we show how to use the stabilizers $P^{(\ell)}_{ij}$ to cancel the long JW strings of non-local hopping and four-fermion interaction terms. The guiding idea is that we will work in the simultaneous $+1$ eigenspace of all relevant $P^{(\ell)}_{ij}$, so that multiplying Hamiltonian terms by these operators leaves all physical matrix elements unchanged while improving locality in the JW representation.

\subsection{Hopping term}

We begin with the hopping term between two arbitrary sites $i<j$. Using~\eqref{eq:JW_a_i} 
\begin{align}
    a_i^\dagger a_j + a_j^\dagger a_i
     &=
    \left(\prod_{n=i}^{j-1} \prod_{k=0}^{\nu} \Z_n^{(k)}\right)
    \frac{2\X^{(0)}_i \X^{(0)}_j
    + 2\Y^{(0)}_i \Y^{(0)}_j}{4}\nonumber\\
     &=
    \frac{1}{2}
    \left(\X^{(0)}_i \X^{(0)}_j
    + \Y^{(0)}_i \Y^{(0)}_j\right)
    \left(\prod_{n=i}^{j-1} \prod_{k=0}^{\nu} \Z_n^{(k)}\right).
    \label{eq:hop_JW}
\end{align}
This is the standard JW form of the hopping term: a two-qubit operator on the physical modes at sites $i$ and $j$ multiplied by a string of $\Z$ operators between $i$ and $j$.

We now consider the transformed operator
\begin{align}
     &(a_i^\dagger a_j + a_j^\dagger a_i)P^{(\ell)}_{ij}\nonumber\\
     &=
    \frac{i}{2}
    \left(\X^{(0)}_i \X^{(0)}_j
    + \Y^{(0)}_i \Y^{(0)}_j\right)
    \left(\prod_{n=i}^{j-1} \prod_{k=0}^{\nu} \Z_n^{(k)}\right)\nonumber\\
     &\quad\times
    \left(\prod_{n'=i}^{j-1} \prod_{k'=0}^{\nu} \Z_{n'}^{(k')}\right)
    \left(\prod_{m=0}^{\ell-1} \Z_i^{(m)} \Z_j^{(m)}\right)
    \X_i^{(\ell)} \Y_j^{(\ell)}.
\end{align}
Since all $\Z$ operators commute and $(\Z)^2 = 1$, the product of the two extensive JW strings collapses to the identity:

\begin{align}
     &(a_i^\dagger a_j + a_j^\dagger a_i)P^{(\ell)}_{ij}\nonumber\\
     &=
    \frac{i}{2}
    \left(\X^{(0)}_i \X^{(0)}_j
    + \Y^{(0)}_i \Y^{(0)}_j\right)
    \left(\prod_{m=0}^{\ell-1} \Z_i^{(m)} \Z_j^{(m)}\right)
    \X_i^{(\ell)} \Y_j^{(\ell)}\nonumber\\
     &=
    -\frac{i}{2}
    \left(\X^{(0)}_i \X^{(0)}_j
    + \Y^{(0)}_i \Y^{(0)}_j\right)
    \left(\prod_{m=1}^{\ell-1} \Z_i^{(m)} \Z_j^{(m)}\right)
    \X_i^{(\ell)} \Y_j^{(\ell)}.
    \label{eq:hop_transformed_local}
\end{align}
Expression~\eqref{eq:hop_transformed_local} is fully local: it acts only on the physical qubits at sites $i$ and $j$ and on the auxiliary qubits with indices $1,\dots,l$ at those sites. The Pauli weight of this operator is
\begin{align}
    w_{\hop}
     &= \underbrace{2}_{\text{physical }(0)\text{ at }i,j}
    + \underbrace{2(\ell-1)}_{\Z_i^{(m)},\Z_j^{(m)},\,1\le m\le \ell-1}
    + \underbrace{2}_{\X_i^{(\ell)},\Y_j^{(\ell)}}\nonumber\\
     &= 2(\ell+1).
\end{align}
Thus the transformation cancels the JW strings of length $\mathcal{O}(N)$, leaving a local term of bounded Pauli weight $w_{\hop}=2(\ell+1)$.

\subsection{Interaction terms}
\label{app:FH_interaction_transformation}
Here we consider Fermi--Hubbard (FH) type density--density interaction $n_in_j$ with $n_i = a_i^\dagger a_i$. For this term no transformation is necessary because $n_i$ is already local in the JW representation.
\begin{equation}
    n_i = a_i^\dagger a_i= \frac{1}{2}(1 + \Z^{(0)}_i).
\end{equation}
The FH density--density term $n_i n_j$ is therefore a product of single-site operators in the JW representation and is already local.

To tackle more general interaction terms, in particular chemistry-style four-fermion terms, we use the identity for $a_i^\dagger P^{(\ell)}_{ij} a_j$, which we now derive.

\subsection{Identity for $a_i^\dagger P^{(\ell)}_{ij} a_j$}
\label{app:aPa_id}
We consider the term $a_i^\dagger P^{(\ell)}_{ij} a_j$ 
\begin{align}
     &a_i^\dagger P^{(\ell)}_{ij} a_j\nonumber\\
     &=
    \left(S_i\,\frac{\X^{(0)}_i + i\Y^{(0)}_i}{2}\right)
    \left(-i\,S_i Z_{i,<\ell} \X^{(\ell)}_i S_j Z_{j,<\ell} \Y^{(\ell)}_j\right)\times\\
    &\times
    \left(S_j\,\frac{\X^{(0)}_j - i\Y^{(0)}_j}{2}\right)\nonumber\\
     &=
    -\frac{i}{4}
    S_i S_i S_j S_j
    (\X^{(0)}_i + i\Y^{(0)}_i)
    Z_{i,<\ell} \X^{(\ell)}_i\nonumber\\
     &\qquad\qquad
    \times Z_{j,<\ell} \Y^{(\ell)}_j
    (\X^{(0)}_j - i\Y^{(0)}_j)\\
    &=-\frac{i}{4}
    (Z_{i,<\ell} Z_{j,<\ell})
    \X^{(\ell)}_i \Y^{(\ell)}_j\nonumber\\
     &\qquad\times
    \left[\X^{(0)}_i \X^{(0)}_j
        + \Y^{(0)}_i \Y^{(0)}_j
        - i(\X^{(0)}_i \Y^{(0)}_j - \Y^{(0)}_i \X^{(0)}_j)\right].
    \label{eq:a_i_dag_P_a_j_final}
\end{align}
It will be the building block for transforming four-fermion (and higher) interaction terms. Note that the Pauli weight of \eqref{eq:a_i_dag_P_a_j_final} is:
\begin{equation}
    w_{a^\dagger Pa} = 2(\ell+1).
    \label{eq:w_aPa}
\end{equation}

\subsection{Transformation of four-fermion terms}
\label{app:four_fermion_trans}

We now consider a generic chemistry-style four-fermion interaction term of the form
\begin{equation}
    h_{ijkl}
    = a_i^\dagger a_j^\dagger a_k a_l
    + a_l^\dagger a_k^\dagger a_j a_i.
\end{equation}
We assume that the indices $i,j,k,l$ are such that the pairs $(i,k)$ and $(j,l)$ correspond to existing stabilizers (if they do not, we add the edges to the interaction graph).

We transform $h_{ijkl}$ by multiplying with $P^{(\ell)}_{ik}P^{(\ell')}_{jl}$:
\begin{align}
    h_{ijkl}
     &\longmapsto
    h_{ijkl}\,P^{(\ell)}_{ik}P^{(\ell')}_{jl}\nonumber\\
     &= \bigl(a_i^\dagger a_j^\dagger a_k a_l
    + a_l^\dagger a_k^\dagger a_j a_i\bigr)
    P^{(\ell)}_{ik}P^{(\ell')}_{jl}.
\end{align}
Reordering the fermionic operators and using the definition of $P^{(\ell)}_{ik}$ and $P^{(\ell')}_{jl}$, this can be written as
\begin{align}
    h_{ijkl}P^{(\ell)}_{ik}P^{(\ell')}_{jl}
     &=
    -(a_i^\dagger P^{(\ell)}_{ik} a_k)
    (a_j^\dagger P^{(\ell')}_{jl} a_l)\nonumber\\
     &\quad
    -\Bigl((a_i^\dagger P^{(\ell)}_{ik} a_k)
    (a_j^\dagger P^{(\ell')}_{jl} a_l)\Bigr)^\dagger.
    \label{eq:hijkl_transform_structure}
\end{align}

From \eqref{eq:a_i_dag_P_a_j_final} and \eqref{eq:w_aPa}, all operators in~\eqref{eq:hijkl_transform_structure} are local: they act only on the physical qubits at sites $i,j,k,l$ and on auxiliary qubits at sites $i,k$ up to level $\ell$ and at sites $j,l$ up to level $\ell'$. The Pauli weight of the transformed four-fermion term is thus just the two contributions of $w_{a^\dagger Pa}$ \eqref{eq:w_aPa}:
\begin{equation}
    w_{4\text{-body}} = 2(\ell+1)+2(\ell'+1) = 4+2\ell+2\ell'
\end{equation}

Regarding the choice of $\ell$ and $\ell'$: the derivation above holds for arbitrary auxiliary indices $\ell$ and $\ell'$. In the stabilizer construction detailed below, each pair $(i,k)$ and $(j,l)$ appearing in such a term is assigned some specific auxiliary layer $\ell$ or $\ell'$, and the auxiliaries are initialized in a joint $+1$ eigenspace of all corresponding $P_{ik}^{(\ell)}$ and $P_{jl}^{(\ell')}$ operators.

Here we show how this four-fermion interaction term is transformed to a qubit term. We define
\begin{align}
    A_{ik}^{(\ell)}  &= a_i^\dagger P^{(\ell)}_{ik} a_k,\\
    B_{jl}^{(\ell')} &= a_j^\dagger P^{(\ell')}_{jl} a_l,
\end{align}
so that~\eqref{eq:hijkl_transform_structure} becomes
\begin{align}
    h_{ijkl}P^{(\ell)}_{ik}P^{(\ell')}_{jl}
     &= -A_{ik}^{(\ell)} B_{jl}^{(\ell')}
    -(A_{ik}^{(\ell)} B_{jl}^{(\ell')})^\dagger.
\end{align}
Using~\eqref{eq:a_i_dag_P_a_j_final} for both $A_{ik}^{(\ell)}$ and $B_{jl}^{(\ell')}$, we obtain
\begin{align}
    A_{ik}^{(\ell)}
     &=
    -\frac{1}{4}
    (Z_{i,<\ell} Z_{k,<\ell})
    \X^{(\ell)}_i \Y^{(\ell)}_k\nonumber\\
     &\quad\times
    \left[\X^{(0)}_i \X^{(0)}_k
        + \Y^{(0)}_i \Y^{(0)}_k
        - i(\X^{(0)}_i \Y^{(0)}_k - \Y^{(0)}_i \X^{(0)}_k)\right],
\end{align}
and similarly
\begin{align}
    B_{jl}^{(\ell')}
     &=
    -\frac{1}{4}
    (Z_{j,<\ell'} Z_{l,<\ell'})
    \X^{(\ell')}_j \Y^{(\ell')}_l\nonumber\\
     &\quad\times
    \left[\X^{(0)}_j \X^{(0)}_l
        + \Y^{(0)}_j \Y^{(0)}_l
        - i(\X^{(0)}_j \Y^{(0)}_l - \Y^{(0)}_j \X^{(0)}_l)\right].
\end{align}
Their product is
\begin{align}
    A_{ik}^{(\ell)} B_{jl}^{(\ell')}
     &=
    \frac{1}{16}
    (Z_{i,<\ell} Z_{k,<\ell})
    (Z_{j,<\ell'} Z_{l,<\ell'})
    \X^{(\ell)}_i \Y^{(\ell)}_k
    \X^{(\ell')}_j \Y^{(\ell')}_l\nonumber\\
     &\quad\times
    \Bigl(   \mathcal{K}_{ik}
    - i\mathcal{L}_{ik}
    \Bigr)
    \Bigl(   \mathcal{K}_{jl}
    - i\mathcal{L}_{jl}
    \Bigr),
\end{align}
where we have defined the physical operators
\begin{align}
    \mathcal{K}_{ik}
     &= \X^{(0)}_i \X^{(0)}_k + \Y^{(0)}_i \Y^{(0)}_k,\\
    \mathcal{L}_{ik}
     &= \X^{(0)}_i \Y^{(0)}_k - \Y^{(0)}_i \X^{(0)}_k,
\end{align}
and similarly for $(j,l)$. Simplifying the above, we obtain:

\begin{align}
    &h_{ijkl}P^{(\ell)}_{ik}P^{(\ell')}_{jl}
     =
    -\frac{1}{8}
    (Z_{i,<\ell} Z_{k,<\ell})
    (Z_{j,<\ell'} Z_{l,<\ell'})\nonumber\\
    &\times\X^{(\ell)}_i \Y^{(\ell)}_k
    \X^{(\ell')}_j \Y^{(\ell')}_l
    \left[\mathcal{K}_{ik}\mathcal{K}_{jl}
        + \mathcal{L}_{ik}\mathcal{L}_{jl}\right].
    \label{eq:hijkl_transformed}
\end{align}
All operators in~\eqref{eq:hijkl_transformed} are local: they act only on the physical qubits at sites $i,j,k,l$ and on auxiliary qubits at sites $i,k$ up to level $\ell$ and at sites $j,l$ up to level $\ell'$.

\subsection{Arbitrary Interaction terms}
\label{app:arbitrary_interaction}
The same approach as in \cref{app:aPa_id} can further be used to eliminate the JW strings from any even fermion term by suitably introducing the stabilizers $P_{ij}^{(\ell)}$. The terms do not necessarily need to be number preserving, the stabilizers allow to eliminate the JW strings in any case. By equivalent arguments as above, a general interaction between $2n$ modes can be transformed to one with weight:
\begin{equation}
    w_{2n} = 2n+\sum_{i=1}^{n} \ell^{(i)},
\end{equation}
where $\ell^{(i)}$ is the auxiliary register on which the stabilizer is defined on.

\subsection{SYK terms}
\label{app:SYK_transform}
A term that appears in the SYK model (see \cref{app:SYK_model}) is the four-Majorana operator $\gamma^{(0)}_i \gamma^{(0)}_j \gamma^{(0)}_k \gamma^{(0)}_l$. Here each $\gamma_i^{(0)}$ can be either the $c^{(0)}_i$ or $d^{(0)}_i$ Majorana operator. We show that it can be transformed to a constant-weight Pauli term by multiplication with the stabilizers. We choose to transform it as:
\begin{equation}
    \gamma^{(0)}_i \gamma^{(0)}_j \gamma^{(0)}_k \gamma^{(0)}_l\xrightarrow[]{}\gamma^{(0)}_i \gamma^{(0)}_j \gamma^{(0)}_k \gamma^{(0)}_l P_{ij}^{(\ell)}P_{kl}^{(\ell')}.
\end{equation}
In principle, the stabilizer assignment is not unique and can be done between any pairs of indices (see discussion in \cref{app:graph_general_H}). However, this assignment does not change the fact that the weight of the Pauli term is independent of the system size. Thus, the transformed term is:
\begin{align}
    &\gamma^{(0)}_i \gamma^{(0)}_j \gamma^{(0)}_k \gamma^{(0)}_l P_{ij}^{(\ell)}P_{kl}^{(\ell')} =\gamma^{(0)}_i c^{(\ell)}_i d^{(\ell)}_j \gamma^{(0)}_j \gamma^{(0)}_k c^{(\ell')}_k d^{(\ell')}_l \gamma^{(0)}_l \nonumber \\&=
    \tilde{\gamma}^{(0)}_i Z_{i,<\ell}\X_i^{(\ell)} Z_{k,<\ell}(-\Y_j^{(\ell)}) \tilde{\gamma}^{(0)}_j \times \\& \times\tilde{\gamma}^{(0)}_k Z_{k,<\ell}\X_k^{(\ell')} Z_{l, <\ell'}(-\Y_l^{(\ell')}) \tilde{\gamma}^{(0)}_l,
\end{align}
where $\tilde{\gamma}_i$ denotes the Majorana representation without the JW string $\X_i^{(0)} (-\Y_i^{(0)})$ when $\gamma_i$ is $c_i^{(0)}$ $(d_i^{(0)})$ respectively. From this transformation one can see that the Pauli weight of the SYK term is:
\begin{equation}
    w_{\mathrm{SYK}} = 4+2\ell+2\ell'.
\end{equation}

\section{Stabilizer Properties}
\label{app:stabilizer_properties}

Here we summarize the relevant identities for the stabilizers $P_{ij}^{(\ell)}$. We saw in \cref{eq:Pij_def} that the stabilizers are defined as:
\begin{equation*}
    P_{ij}^{(\ell)} = ic_i^{(\ell)}d_j^{(\ell)}.
\end{equation*}

Firstly, we show that the stabilizers square to an identity:
\begin{equation}
    P_{ij}^{(\ell)2} = ic_i^{(\ell)}d_j^{(\ell)}ic_i^{(\ell)}d_j^{(\ell)} = c_i^{(\ell)}c_i^{(\ell)}d_j^{(\ell)}d_j^{(\ell)}=\Id,
\end{equation}
where we have used the Majorana algebra relations \cref{eq:majorana_id}. Since it squares to an identity, this implies that the eigenvalues of the stabilizer are $\pm 1$. 
The respective eigenstates are:
\begin{equation}
\label{eq:app_stab_eigenstates}
    P_{ij}^{(\ell)}\frac{c_i^{(\ell)} \mp id_j^{(\ell)}}{\sqrt{2}}\ket{0_{\aux}} = \pm1 \frac{c_i^{(\ell)} \mp id_j^{(\ell)}}{\sqrt{2}}\ket{0_{\aux}}.
\end{equation}
Secondly, we verify that:
\begin{equation}
    P_{ij}^{(\ell)\dagger} = (ic_i^{(\ell)}d_j^{(\ell)})^\dagger = -id_j^{(\ell)}c_i^{(\ell)} = ic_i^{(\ell)}d_j^{(\ell)} = P_{ij}^{(\ell)},
\end{equation}
the stabilizers are indeed hermitian. Once again we have used the \cref{eq:majorana_id}.

Third, we want to formally establish when different stabilizers commute:
\begin{align}
\label{eq:P_comm}
    &\comm{P_{ij}^{(\ell)}}{P_{kl}^{(\ell')}} = \comm{c_k^{(\ell')}d_l^{(\ell')}}{c_i^{(\ell)}d_j^{(\ell)}} \nonumber\\&= c_i^{(\ell)}c_k ^{(\ell')}d_j^{(\ell)}d_l^{(\ell')} -c_k ^{(\ell')}c_i^{(\ell')}d_l^{(\ell')}d_j^{(\ell)}\nonumber\\
    &=c_i^{(\ell)}c_k ^{(\ell')}d_j^{(\ell)}d_l^{(\ell')} - (2\delta_{ik}\delta_{\ell\ell'}-c_i^{(\ell)}c_k ^{(\ell')})(2\delta_{jl}\delta_{\ell\ell'}-d_j^{(\ell)}d_l^{(\ell')})\nonumber\\
    &=2\delta_{\ell\ell'}(c_i^{(\ell)}c_k ^{(\ell)}\delta_{jl}+\delta_{ik}d_j^{(\ell)}d_l^{(\ell)}-2\delta_{ik}\delta_{jl}),
\end{align}
where going to the third line we have used \cref{eq:majorana_id}.
The \cref{eq:P_comm} shows that if $\ell\neq\ell'$, the stabilizers always commute, as well as when both $i\neq k, j\neq l$ on the same auxiliary register.

\section{Grouping of the stabilizers using graph coloring}
\label{app:stabilizer_grouping}

In this section, we focus on the task of generating a set of commuting stabilizers $P^{(\ell)}_{ij}$, given an arbitrary Hamiltonian.

\subsection{Constructing the interaction graph}
\label{app:graph_general_H}

Firstly, it is trivial to generate a hopping Hamiltonian. Given a Hamiltonian $H_{\mathrm{\hop}}$ on $N$ modes we construct a graph $\mathcal{G} = (V,E)$ with $\abs{V} = N$. The graph $\mathcal{G}$ encodes the locality pattern: we assume that $(i,j)\in E$ exactly when the hopping term between sites $i$ and $j$ is present ($t_{ij}\neq0$).

However, one might also be interested in higher order interaction terms like the four-fermion Hamiltonian:
\begin{equation}
    H_{4} = \sum_{ijkl} U_{ijkl}a_i^\dagger a_j^\dagger a_k a_l.
\end{equation}
Since this is a four-body interaction, there is no notion of an unique interaction edge. However, as we saw from \cref{app:four_fermion_trans,app:arbitrary_interaction} any even fermionic term on $2n$ modes can be simplified using $n$ stabilizers. In the case of these higher order interactions there is freedom on choosing the stabilizers we wish to use. For example, when simplifying the four-fermion term:
\begin{equation*}
    a_i^\dagger a_j^\dagger a_k a_l,
\end{equation*}
one can use stabilizers between edges in various combinations. In particular, there are $\binom{4}{2}=6$ different ways to choose between which vertices to insert the stabilizers. Just to illustrate two of them, both multiplication with $P_{ik}P_{jl}$ or $P_{il}P_{kj}$ lead to a local qubit term. 

Thus, to incorporate terms like $H_4$ in the interaction graph $\mathcal{G}$ we:
\begin{enumerate}
    \item pick one stabilizer configuration that simplifies the higher order terms
    \item for each stabilizer $P_{ij}^{(\ell)}$ that we require, we add an edge $(i,j)$ to $E$, if it is not already present.
\end{enumerate}
This strategy works for creating an appropriate (albeit not unique) interaction graph. We do not concern ourselves with finding the best possible solution here. In the next section, we assume that we are given an instance of such an interaction graph and show how to group the stabilizers across the minimal number of auxiliary modes, ensuring that they all mutually commute.

\subsection{Edge coloring and auxiliary layers}
\label{app:edge_coloring}

To assign commuting stabilizers, we use a proper \emph{edge coloring} of the graph $\mathcal{G}$. An edge coloring is a map
\begin{equation}
    \mathrm{col}:E\to\{1,\dots,\chi\}
\end{equation}
such that no two edges incident on the same vertex share the same color. The minimum number $\chi$ of colors for which such a coloring exists is the \emph{chromatic index} (or edge chromatic number) of $\mathcal{G}$.

For each color $\gamma\in\{1,\dots,\chi\}$ we denote by $E_{\gamma}\subseteq E$ the set of edges of color $\gamma$,
\begin{equation}
    E_{\gamma} = \{(i,j)\in E\mid \mathrm{col}(i,j) = \gamma\}.
\end{equation}
Because $\mathrm{col}$ is a proper edge coloring, each $E_{\gamma}$ is a \emph{matching}: no two edges in $E_{\gamma}$ share a common endpoint.

We first imagine that we have at least one auxiliary mode per color, that is, $\nu\ge \chi$, and that we temporarily assign a distinct auxiliary index $\ell=\gamma$ to each color class $E_{\gamma}$. For each edge $(i,j)\in E_{\gamma}$ we choose an arbitrary orientation and define a stabilizer
\begin{equation}
    P_{ij}^{(\gamma)} = i\,c_i^{(\gamma)} d_j^{(\gamma)}.
\end{equation}
The following lemma makes the commutation properties of each color class precise.

\begin{lemma}\label{lemma:P_commute}
    Fix a color $\gamma$ and an auxiliary index $\ell$. Let
    \begin{equation}
        \mathcal{P}_{\gamma}^{(\ell)} = \{P^{(\ell)}_{ij}\}_{(i,j)\in E_{\gamma}}
    \end{equation}
    be the set of operators defined as above from the edges of $E_{\gamma}$ and an arbitrary orientation. Then the operators in $\mathcal{P}_{\gamma}^{(\ell)}$ mutually commute.
\end{lemma}
\begin{proof}
    Note that $E_{\gamma}$ is a matching, so any two distinct edges $(i,j)$ and $(k,l)$ in $E_{\gamma}$ are vertex-disjoint:
    \begin{equation}
        \{i,j\}\cap\{k,l\} = \emptyset.
    \end{equation}
    Therefore, from \cref{eq:P_comm} the respective stabilizers commute. Since this holds for any pair of distinct operators in $\mathcal{P}_{\gamma}^{(\ell)}$, all of them mutually commute.
\end{proof}

Next we show that we can \emph{reuse} ancillas across colors: stabilizers from any two colors can be encoded in the same auxiliary index $\ell$ while retaining commutativity. This allows us to reduce the number of auxiliary modes per site by almost a factor of two.

\begin{lemma}\label{lemma:half_ancillas}
    Let $\gamma$ and $\gamma'$ be two colors with edge sets $E_{\gamma}$ and $E_{\gamma'}$. Consider the graph
    \begin{equation}
        \mathcal{G}_{\gamma\gamma'} = (V, E_{\gamma}\cup E_{\gamma'}).
    \end{equation}
    Then there exists an orientation of all edges in $E_{\gamma}\cup E_{\gamma'}$ such that, if we fix a single auxiliary index $\ell$ and define
    \begin{equation}
        P^{(\ell)}_{ij} = i c_i^{(\ell)} d_j^{(\ell)}
    \end{equation}
    for every oriented edge $(i,j)\in E_{\gamma}\cup E_{\gamma'}$, all these stabilizers commute mutually. In particular, the two color classes $E_{\gamma}$ and $E_{\gamma'}$ can be represented using the same layer $\ell$ of auxiliary fermions.
\end{lemma}
\begin{proof}
    Because $E_{\gamma}$ and $E_{\gamma'}$ are both matchings, the degree of any vertex in $\mathcal{G}_{\gamma\gamma'}$ is at most two:
    \begin{equation}
        \deg_{\mathcal{G}_{\gamma\gamma'}}(v) \le 2\quad\text{for all }v\in V.
    \end{equation}
    Therefore, $\mathcal{G}_{\gamma\gamma'}$ is a disjoint union of paths and cycles. On each connected component (path or cycle) we choose an orientation that is \emph{consistent along the component}: for a path $(v_1,v_2,\dots,v_r)$ we orient all edges as
    \begin{equation}
        v_1 \to v_2,\quad v_2 \to v_3,\quad \dots,\quad v_{r-1} \to v_r,
    \end{equation}
    and for a cycle $(v_1,v_2,\dots,v_r,v_1)$ we orient all edges, for example, as
    \begin{equation}
        v_1 \to v_2,\quad v_2 \to v_3,\quad \dots,\quad v_{r-1} \to v_r,\quad v_r \to v_1.
    \end{equation}
    With this choice, at every vertex $v$ there is at most one outgoing and at most one incoming edge in $E_{\gamma}\cup E_{\gamma'}$. For each, now oriented, edge $(i,j)$ in $E_{\gamma}\cup E_{\gamma'}$, we now define the stabilizer $P^{(\ell)}_{ij} = i c_i^{(\ell)} d_j^{(\ell)}$. Due to the orientation of the operators all $P^{(\ell)}_{ij}$ are products of distinct Majorana terms, hence the elements in the set
    \begin{equation}
        \mathcal{P}_{\gamma\gamma'}^{(\ell)}
        = \{P^{(\ell)}_{ij}\}_{(i,j)\in E_{\gamma}\cup E_{\gamma'}}
    \end{equation}
    commute pairwise (see \cref{eq:P_comm}).
\end{proof}

As an immediate corollary, we can pack at most two colors into each auxiliary layer without destroying commutativity. If $\chi$ is the chromatic index of $\mathcal{G}$, we can therefore choose
\begin{equation}
    \nu = \left\lceil \frac{\chi}{2} \right\rceil
\end{equation}
auxiliary modes per site, and assign colors to layers as
\begin{equation}
    \text{layer }\ell\text{ carries colors }2\ell-1\text{ and }2\ell,
\end{equation}
with the understanding that if $2\ell>\chi$ then color $2\ell$ is absent. For each layer $\ell$ we orient all edges of colors $2\ell-1$ and $2\ell$ as in Lemma~\ref{lemma:half_ancillas}, and define stabilizers $P_{ij}^{(\ell)}$ for all such oriented edges.

\subsection{Regular graph properties}
\label{app:regular_graphs}
A relevant example of sparse interactions graphs are the $d-$regular graphs. A graph $\mathcal{G} = (V,E)$ is $d-$ regular if each $v\in V$ belong to exactly $d$ edges in $E$. 

For regular graphs, it is possible to bound their chromatic index $\chi$ using the Vizing's theorem \cite{vizing1964estimate}:
\begin{equation}
    d\leq\chi \leq d+1.
\end{equation}

\section{Ordered state circuit construction}
\label{app:ordered_state_prep}

In \cref{sec:aux_state_preparation} we saw that the initial state can be constructed by application of $\mathcal{C}^{(\ell)}$. To implement this operation, we first permute the fermions to their ordered form with $U_{\sigma_{2\ell}}$ (or $U_{\sigma_{2\ell-1}}$), followed by an \emph{ordered} state preparation $\mathcal{C}_{\mathrm{ord}}^{(\ell)}$ (see \cref{eq:C_ord}). We discuss the implementation of $\mathcal{C}_{\mathrm{ord}}^{(\ell)}$ here.

\subsection{Mapping to a unitary}
\label{sec:unitary_ord}

First, note that the operation $\frac{c^{(\ell)}_i-id^{(\ell)}_j}{\sqrt{2}}$ is not unitary. To embed it into a unitary we introduce an additional auxiliary qubit:
\begin{align}
    \mathcal{U}^{(\ell)}_{ij} = \ketbra{0}_{\anc}\otimes c^{(\ell)}_i+\ketbra{1}_{\anc}\otimes (-id^{(\ell)}_j).
\end{align}

One can verify that 
\begin{align*}
&\mathcal{U}^{(\ell)\dagger}_{ij}\mathcal{U}^{(\ell)}_{ij} = \mathcal{U}^{(\ell)}_{ij}\mathcal{U}^{(\ell)\dagger}_{ij}= \\&= \ketbra{0}_{\anc}\otimes c^{(\ell)2}_i+\ketbra{1}_{\anc}\otimes d^{(\ell)2}_j\\&=\ketbra{0}_{\anc}\otimes \Id + \ketbra{1}_{\anc} \otimes \Id = \Id,
\end{align*}
$\mathcal{U}^{(\ell)}_{ij}$ is indeed unitary. To prepare the desired eigenstate of $P_{ij}^{(\ell)}$ we perform:
\begin{align}
    &\Had_{\anc} \mathcal{U}^{(\ell)}_{ij} \Had_{\anc} \ket{0_{\anc}}\ket{\Psi_{\aux}}=\\&= \Had_{\anc}\left(\frac{\ket{0_{\anc}}}{\sqrt{2}}c_i^{(\ell)}-\frac{\ket{1_{\anc}}}{\sqrt{2}}id_j\right)\ket{\Psi_{\aux}}\nonumber\\
    &=\left[\frac{\ket{0_{\anc}}}{2}(c_i^{(\ell)}-id_j^{(\ell)})+\frac{\ket{1_{\anc}}}{2}(c_i^{(\ell)}+id_j^{(\ell)})\right]\ket{\Psi_{\aux}}\nonumber.
\end{align}
Following this circuit, we measure the ancillary register:
\begin{itemize}
    \item $\ket{0_{\anc}}$ indicates that we have a $+1$ eigenstate of $P_{ij}^{(\ell)}$
    \item $\ket{1_{\anc}}$ indicates that we have a $-1$ eigenstate of $P_{ij}^{(\ell)}$
\end{itemize}
which can be easily verified from \cref{eq:app_stab_eigenstates}. One option to ensure that we have the correct eigenstate would be to post-select the outcome on $\ket{0_{\anc}}$, however we argue that it isn't even necessary. One can simply notice that:
\begin{equation}
    (-P_{ij}^{(\ell)})\frac{c^{(\ell)}_i+id^{(\ell)}_j}{\sqrt{2}}\ket{0_{\aux}} = \frac{c^{(\ell)}_i+id^{(\ell)}_j}{\sqrt{2}}\ket{0_{\aux}}.
\end{equation}
Hence, the outcome $\ket{1_{\anc}}$ also indicates that we have prepared the $+1$ eigenstate of the $-P_{ij}^{(\ell)} = -ic^{(\ell)}_id^{(\ell)}_j$. Thus, in the case of the $\ket{1_{\anc}}$ outcome we can simply redefine the respective stabilizer to be $P_{ij}^{(\ell)}\xrightarrow[]{} -ic^{(\ell)}_id^{(\ell)}_j$. Note that this simple change of sign does not affect the commutation properties of the stabilizers, nor the weight of the transformed Hamiltonian terms. It will only affect the sign in front of the terms, whose transformation feature the respective stabilizer.

\subsection{Circuit decomposition of the ordered state preparation}

Having showed how we can initialize the respective state, let us show how we can prepare an efficient circuit in the ordered case. To try that, let us first try to simplify $\mathcal{C}^{(\ell)}_{\ord}$ (\cref{eq:C_ord}) in the qubit picture:
\begin{align}
\label{eq:C_ord_simpl_1}
    &2^{\frac{N}{4}}\mathcal{C}^{(\ell)}_{\ord}= \prod_{k=1}^{N/2}(S_{2k-1}Z_{2k-1,<\ell} \X^{(\ell)}_{2k-1}+iS_{2k}Z_{2k,<\ell} \Y^{(\ell)}_{2k})\nonumber\\
    &=\prod_{k=1}^{N/2}S_{2k-1}(Z_{2k-1,<\ell} \X^{(\ell)}_{2k-1}+iZ_{2k-1,\leq\nu}Z_{2k,<\ell} \Y^{(\ell)}_{2k})\\
    &=\prod_{k=1}^{N/2}(Z_{2k-1,<\ell} \X^{(\ell)}_{2k-1}+iZ_{2k-1,\leq\nu}Z_{2k,<\ell} \Y^{(\ell)}_{2k})\prod_{k=1}^{N/2}S_{2k-1}.\nonumber
\end{align}
Notice that in the first product, the terms are local and only supported on the sites $2k-1,2k$. To fully simplify it to a product of local contributions, we need to simplify the product $\prod_{k=1}^{N/2}S_{2k-1}$:
\begin{align}
    \prod_{k=1}^{N/2}S_{2k-1} = \prod_{k=1}^{N/2}(Z_{2k-1,\leq\nu})^{N/2+1-k}(Z_{2k,\leq\nu})^{N/2-k},
\end{align}
where we have used the fact that any $\Z^2=1$. Plugged back into \cref{eq:C_ord_simpl_1} it yields:
\begin{align}
    \label{eq:C_ord_simpl_2}
    &2^{\frac{N}{4}}\mathcal{C}^{(\ell)}_{\ord}=\prod_{k=1}^{N/2}\Bigg(Z_{2k-1,<\ell} \X^{(\ell)}_{2k-1}\nonumber\\&+iZ_{2k-1,\leq\nu}Z_{2k,<\ell} \Y^{(\ell)}_{2k}\Bigg)\prod_{k=1}^{N/2}S_{2k-1}\nonumber\\
    &=\prod_{k=1}^{N/2} \Bigg(\X^{(\ell)}_{2k-1}Z_{2k-1,<\ell}  (Z_{2k-1,\leq\nu})^{N/2+1-k}\nonumber\\&+i\Y^{(\ell)}_{2k}Z_{2k-1,\leq\nu}Z_{2k,<\ell}(Z_{2k,\leq\nu})^{N/2-k}\Bigg)\nonumber\\
    &=\prod_{k=1}^{N/2} \Bigg(\tilde{c}_{2k-1}^{(\ell)}-i\tilde{d}_{2k}^{(\ell)}\Bigg),
\end{align}
where at the end we have introduced simplified, transformed operators:
\begin{align}
    \label{eq:tilde_ops}
    \tilde{c}^{(\ell)}_{2k-1}&=\X^{(\ell)}_{2k-1}Z_{2k-1,<\ell}  (Z_{2k-1,\leq\nu})^{N/2+1-k},\\
    \tilde{d}^{(\ell)}_{2k}&=-\Y^{(\ell)}_{2k}Z_{2k-1,\leq\nu}Z_{2k,<\ell}(Z_{2k,\leq\nu})^{N/2-k}.
\end{align}
Notice that in \cref{eq:C_ord_simpl_2}, the $\mathcal{C}^{(\ell)}_{\ord}$ has the same form as originally in \cref{eq:C_ord}, just now with $\tilde{c}_i^{(\ell)}, \tilde{d}_j^{(\ell)}$ that have simplified, local \emph{JW strings}. Furthermore, notice that the Pauli weights of these terms are:
\begin{align}
    \label{eq:weight_tilde_ops}
    w(\tilde{c}^{(\ell)}_{2k-1})&= \min(\ell, \nu-\ell)\leq \nu,\\
    w(\tilde{d}^{(\ell)}_{2k})&= \min(\ell+\nu, 2\nu-\ell)\leq 2\nu.
\end{align}
Thus, we can proceed to define the circuit $\mathcal{U}_{\ord}$ that will implement $\mathcal{C}^{(\ell)}_{\ord}$ as:
\begin{align}
    \label{eq:U_ord_circ}
    \mathcal{U}_{\ord} &= \prod_{k=1}^{N/2}\tilde{\mathcal{U}}_k,\\
    \text{where}\;&\tilde{\mathcal{U}}_k = \ketbra{0}_{\anc,k}\tilde{c}^{(\ell)}_{2k-1}+\ketbra{1}_{\anc,k}(-i\tilde{d}^{(\ell)}_{2k}).
\end{align}
Note that to implement this, we need to introduce additional $N/2$ ancillary qubits. Furthermore, the key of this transformation is that different $\tilde{\mathcal{U}}_k$ act on disjoint qubits. Thus, we can implement all $\tilde{\mathcal{U}}_k$ in parallel.

Finally, to establish the circuit depth required to prepare the ordered auxilliary state we need to establish the depth required to implement any $\tilde{\mathcal{U}}_k$. To do that, notice the structure of the terms in \cref{eq:tilde_ops}. It consists of a product of $\Z$ operators, along with a $\X$ ($i\Y$) Pauli operator. The circuit $\tilde{\mathcal{U}}_k$ is simply their controlled version. Thus, the most expensive part of this circuit will be to implement the control-$\Z$ ($\CZ$) operations, controlled by the ancillary register, acting on all the relevant $\Z$ operators. Importantly, such circuits are called $\CZ$ \emph{fanout} circuits. It has been shown, that a $\CZ$ fanout on $n$ qubits can be implemented in a $\log(n)$ circuit depth \cite{Remaud_2025, moore1999quantumcircuitsfanoutparity, fang2003quantumlowerboundsfanout}. The circuit for it would look like:
\begin{align}
    \tilde{\mathcal{U}}_k &= \X_{\anc} \CX^{(\ell)}_{\anc,2k-1} \prod_{t \in \mathcal{Z}_k^c} \left(\CZ_{\anc, t}\right) \X_{\anc}\times \\
    &\times \mathrm{C}(i\Y^{(\ell)})_{\anc,2k} \prod_{t \in \mathcal{Z}_k^d} \left(\CZ_{\anc, t}\right),
\end{align}
where the sets $\mathcal{Z}_k^c$ ($\mathcal{Z}_k^d$) contain the information of the respective target $\Z$ operator indices from \cref{eq:tilde_ops} that are involved in the \emph{fanout}. Note that for implementation $i\Y^{(\ell)}_{2k} = \Z^{(\ell)}_{2k}\X^{(\ell)}_{2k}$.
Thus, also $\tilde{\mathcal{U}}_k$ can be implemented in $\mathcal{O}(\log(\nu))$ depth. Note that one needs to take the same steps with regards to measuring the ancilla and adjusting the sign of the stabilizer as in \cref{sec:unitary_ord}. To sum it up, the operation $\mathcal{C}^{(\ell)}_{\ord}$ can be embedded into a unitary operation and realized as a quantum circuit of depth $\mathcal{O}(\log(\nu))$ along with $N/2$ additional ancillas.

\section{Trotter Implementation}
\label{app:Trotter_errors}

\subsection{Commutator bounds}

Consider the commutator bounds on the transformed Hamiltonian:
\begin{equation}
    \comm{\tilde{h_{ij}}}{\tilde{h_{kl}}} = \comm{h_{ij}P^{(\ell)}_{ij}}{h_{kl}P_{kl}^{(\ell')}} = \comm{h_{ij}}{h_{kl}}P^{(\ell)}_{ij}P_{kl}^{(\ell')},
\end{equation}
since the projectors commute with themselves (our construction) and the physical operators. This also straightforwardly extends to higher order fermion interactions. Thus,
\begin{align}
    \norm{\comm{\tilde{h_{ij}}}{\tilde{h_{kl}}}}_2 &= \norm{\comm{h_{ij}}{h_{kl}}P^{(\ell)}_{ij}P_{kl}^{(\ell')}}_2 = \norm{\comm{h_{ij}}{h_{kl}}}_2,
\end{align}
since $P_{ij}^{(\ell)\dagger}P_{ij}^{(\ell)}= \Id$.

The total nested commutator for the first-order Trotter evolution is:
\begin{align}
    \Lambda &= \sum_{\gamma = 1}^{\chi} \sum_{\delta = \gamma+1}^\chi \norm{\comm{\tilde{h}_{\gamma}}{\tilde{h}_{\delta}}}_2 = \\
    &=\sum_{\gamma = 1}^{\chi} \sum_{\delta = \gamma+1}^\chi \norm{\comm{h_{\gamma}}{h_{\delta}}}_2 = \\
    &=\sum_{\gamma = 1}^{\chi} \sum_{\delta = \gamma+1}^\chi \mathcal{O}(N) = \mathcal{O}(\chi^2N),
\end{align}
where we assumed that between any two color sets, each term in the set overlaps with only at most two others and that the norm of all couplings have been normalized to $\mathcal{O}(1)$.

\subsection{Implementing $n-$qubit evolution}
\label{sec:Pauli_gadget}

Consider an evolution of a Hamiltonian term supported on $n-$qubits. Such an evolution can be implemented with a single qubit rotation, conjugated between a $n-$qubit $\mathrm{CNOT}$ staircase (Pauli gadget method). Such a $\mathrm{CNOT}$ ladder can be implemented with a $O(\log(n))$ depth circuit without any ancillas \cite{Remaud_2025}. Thus, any evolution of a $n-$qubit Pauli term can be done in a $O(\log(n))$ depth.

\section{Example models}
\label{app:example_models}

\subsection{Fermi--Hubbard model}

Consider a Fermi--Hubbard model on a $d-$regular graph $\mathcal{G} = (V,E)$, meaning that each vertex $v \in V$ has exactly $d$ edges:
\begin{align}
    \label{eq:FH_hamiltonian}
    H_{\mathrm{FH}} &= H_{\mathrm{hop}} + H_{\mathrm{int}} = \\
    &=\sum_{(i,j)\in E} t_{ij}\,(a_i^\dagger a_j + a_j^\dagger a_i) + \sum_{(i,j)\in E} V_{ij} a_i^\dagger a_i a_j^\dagger a_j.
\end{align}
On a $d-$regular graph $d\leq\chi\leq d+1$ (see \cref{app:regular_graphs}), thus we need at most $\nu = \lceil (d+1)/2\rceil$ ancillary fermions per site. Even though the interaction part of the \cref{eq:FH_hamiltonian} is a four-fermion term, it already has a local support as we saw in \cref{app:FH_interaction_transformation}:
\begin{equation*}
    H_{\mathrm{int}} = \sum_{(i,j)\in E} V_{ij} a_i^\dagger a_i a_j^\dagger a_j = \sum_{(i,j)\in E} \frac{V_{ij}}{4}(\Id + \Z_i^{(0)})(\Id+\Z_j^{(0)}),
\end{equation*}
thus requiring no transformation. Therefore, the initial auxiliary state preparation proceeds identically to the one for $H_{\hop}$ as discussed in the main text, since we do not require any stabilizers for the $H_{\mathrm{int}}$. 

Finally, when performing the Trotterized time evolution:
\begin{align}
    &W_{\mathrm{FH}} = e^{-i\tau H_{\mathrm{FH}}} =\\&= \sum_{\gamma=1}^\chi \prod_{\alpha=1}^\chi \left( e^{-i\tau \tilde{h}_{\alpha}^{\mathrm{hop}}} e^{-i\tau \tilde{h}_{\alpha}^{\mathrm{int}}}\right) + \mathcal{O}(\Lambda \tau^2),
\end{align}
where note that $[\tilde{h}^{\mathrm{hop}}_\alpha,\tilde{h}^{\mathrm{int}}_\alpha] = 0$ and $\Lambda=\sum_{\alpha=1}^\chi \sum_{\beta =\alpha+1}^\chi \norm{[\tilde{h}^{\mathrm{hop}}_\alpha+\tilde{h}^{\mathrm{hop}}_\alpha, \tilde{h}^{\mathrm{hop}}_\beta+\tilde{h}^{\mathrm{hop}}_\beta]}_2$, hence $\Lambda = \mathcal{O}(\chi^2N)$. Thus, the depth for implementing a single-color Trotter layer is still $O(\log(\chi))$, and the full cost for implementing $M$ steps of the Trotterized time evolution is:
\begin{align}
    \mathrm{Depth}(W_{\mathrm{FH}}^M) &= \mathrm{Depth}_{\init} + M\mathrm{Depth}_{\mathrm{step}}\\
    &=\mathcal{O}(d\log(d) + d\log( N)+Md\log(d))\nonumber.
\end{align}

\subsection{SYK model}
\label{app:SYK_model}
We consider a quartic Sachdev--Ye--Kitaev (SYK) type Hamiltonian \cite{Maldacena_2016} acting on $N$ Majorana fermions $\{\gamma_i\}_{i=1}^N$,
\begin{equation}
H_{\mathrm{SYK}} = \sum_{e \in \mathcal{E}} J_e \prod_{i \in e} \gamma_i ,
\end{equation}
where $\mathcal{E}$ denotes the edge set of a 4-uniform hypergraph, such that each interaction term involves exactly four fermionic operators. No assumption is made on the distribution of the couplings $J_e$, which may be arbitrary. We restrict to a model defined on a 4-uniform $d$-regular hypergraph, so that each Majorana fermion participates in exactly $d$ quartic interaction terms. This defines a sparse SYK$_4$-type model, in contrast to the fully connected SYK$_4$ Hamiltonian, where the number of interaction terms per Majorana fermion scales as $\mathcal{O}(N^3)$. In the following, each $\gamma_i$ is taken to be either the $c^{(0)}_i$ or $d^{(0)}_i$ Majorana operator. Note that for $N$-Majorana fermion modes it is sufficient to choose $N/2$ fermions.

The first challenge in simulating the model is to construct a stabilizer graph $\mathcal{G}_{\mathrm{stab}}$ from the underlying 4-uniform hypergraph $\mathcal{G}$. A simple choice is to split each hyperedge $e = (i,j,k,l) \in \mathcal{E}$ into two edges $(i,j)$ and $(k,l)$ in the stabilizer graph $\mathcal{G}_{\mathrm{stab}}$. Since each Majorana mode participates in exactly $d$ hyperedges, this construction ensures that it also participates in exactly $d$ edges in $\mathcal{G}_{\mathrm{stab}}$. Since there are two Majorana modes associated with each physical site, this graph has degree at most $2d$, thus we need $\nu = d$ ancillas per fermion mode. In \cref{app:SYK_transform}, we show that such a splitting allows the stabilizers $P_{ij}^{(\ell)}$ to cancel the \emph{JW strings} appearing in the terms of $H_{\mathrm{SYK}}$.
 In \cref{app:regular_graphs}, we show that the chromatic index of a regular graph is directly proportional to its degree, hence the initial state preparation of the auxilliary modes can be implemented with depth
\begin{equation}
    \mathrm{Depth}_{\init} = \mathcal{O}\!\left(d\log(d) + d\log(N)\right).
\end{equation}
Furthermore, the transformed interaction terms have Pauli weight scaling as $w_{\mathrm{SYK}} = \mathcal{O}(d)$ (\cref{app:SYK_transform}). To determine the full depth of a single Trotter step, we require the chromatic index of the hypergraph $\mathcal{G}$, since this bounds the number of Trotter layers into which the terms can be grouped. In \cite{pippenger_hypergraph}, it was shown that the chromatic index of a $d$-regular hypergraph also scales as $\mathcal{O}(d)$. Consequently, the depth of a full Trotter step is
\begin{equation}
\mathrm{Depth}_{\mathrm{step}}\!\left(W_{\mathrm{SYK}}\right) = \mathcal{O}\!\left(d\log(d)\right).
\end{equation}

\end{document}